\documentclass{eptcs}
\providecommand{\event}{LSFA 2024} % Name of the event you are submitting to

\usepackage{iftex}

\ifpdf
  \usepackage{underscore}         % Only needed if you use pdflatex.
  \usepackage[T1]{fontenc}        % Recommended with pdflatex
\else
  \usepackage{breakurl}           % Not needed if you use pdflatex only.
\fi

\usepackage{cutelmet}
\usepackage{amsmath}  
\usepackage{amsthm}
\usepackage{amssymb}    
\usepackage{mathtools}
\usepackage{bussproofs}
\usepackage{proof}
\usepackage{cite}
\usepackage{xcolor}

\newtheorem{definition}{Definition}

\newtheorem{lemma}{Lemma}
\newtheorem{example}{Example}

\title{Towards an Analysis of Proofs in Arithmetic}
\author{Alexander Leitsch
\institute{Institute of Logic and Computation,\\TU Wien, 
   Vienna, Austria\\}
\email{leitsch@logic.at}
\and
Anela Loli\'c \thanks{Partially supported by FWF project I-5848-N and recipient of an APART-MINT Fellowship of the Austrian Academy of Sciences at the Institute of Logic and Computation of the TU Wien.}
\institute{ Kurt G\"odel Society} 
\institute{Institute of Logic and Computation,\\TU Wien, 
   Vienna, Austria\\}
\email{anela@logic.at}
\and 
Stella Mahler \thanks{Partially supported by FWF project I-5848-N.}
\institute{Institute of Logic and Computation,\\TU Wien, 
   Vienna, Austria\\}
\email{stella@logic.at}
}

\def\titlerunning{Towards an Analysis of Proofs in Arithmetic}
\def\authorrunning{Leitsch, Loli\'c, Mahler}
\begin{document}
\maketitle

\begin{abstract}
Inductive proofs can be represented as proof schemata, i.e. as parameterized sequences of proofs defined in a primitive recursive way. Applications of proof schemata can be found in the area of automated proof analysis where the schemata admit (schematic) cut-elimination and the construction of Herbrand systems.
This work focuses on the expressivity of proof schemata. We show that proof schemata can simulate primitive recursive arithmetic. The translation of proofs in arithmetic to proof schemata can be considered as a crucial step in the analysis of inductive proofs. 
\end{abstract}

\section{Introduction}

Mathematical induction is one of the most important principles in real mathematics, thus any substantial and relevant approach to analyze mathematical proofs has to take into account induction. But in systems with induction rules, essential proof theoretic concepts and transformation become problematic. In particular Gentzen's method of cut-elimination fails for general induction proofs and Herbrand's theorem cannot be realized \cite{Takeuti}. The reason is that in the cut-elimination method \`a la Gentzen cuts cannot be shifted over induction rules. We will illustrate this problem on a concrete example below.
\begin{example}\label{ex.PRA-proof}
Let $P(x)$ denote $h(x)=0$ for a primitive recursive function defined in primitive recursive arithmetic {\PRA}, and $g(x,y)$ be any primitive recursive function in {\PRA}. Moreover, let $f$ be a unary function and
$$\Ecal = \{f(x,0) = x, \quad f(x,s(y)) = s(f(x,y))\}$$
(defining addition). Now we consider the sequent
$$S\colon \forall x (P(x) \to P(s(x))) \vdash \forall n \forall x ((P(f(x,n)) \to P(g(x,n))) \to (P(x) \to P(g(x,n)))).$$
$S$ is not valid in pure first-order logic and does not have a Herbrand sequent w.r.t. to the theory $\Ecal$ and hence cannot be proven without induction. Therefore, there is no proof of $S$ in pure first-order logic. We need the following inductive lemma
$$\forall x (P(x) \to P(s(x))) \vdash \forall n \forall x (P(x) \to P(f(x,n))).$$
A proof $\varphi$ of this lemma is
\begin{small}
\begin{prooftree} 
	\AxiomC{$(\varphi_1)$}
	\noLine
	\UnaryInfC{$\vdash \forall x P(x) \to P(f(x,0)))$}
	\AxiomC{$(\varphi')$}
	\noLine
	\UnaryInfC{$\forall x (P(x) \to P(s(x))), \forall x (P(x) \to P(f(x,0))) \vdash \forall n \forall x (P(x) \to P(f(x,n)))$}
	\RightLabel{$cut$}
	\BinaryInfC{$\forall x (P(x) \to P(s(x))) \vdash \forall n \forall x (P(x) \to P(f(x,n)))$}
\end{prooftree}
\end{small}
where $\varphi'$ is
\begin{scriptsize}
\begin{prooftree} 
	\AxiomC{$(\varphi_2)$}
	\noLine
	\UnaryInfC{$\forall x (P(x) \to P(f(x,0))) \vdash \forall x (P(x) \to P(f(x,0)))$}
	\AxiomC{$(\varphi_3)$}
	\noLine
	\UnaryInfC{$\forall x (P(x) \to P(s(x))), \forall x (P(x) \to P(f(x,k))) \vdash \forall x (P(x) \to P(f(x,s(k))))$}
	\RightLabel{$ind$}
	\BinaryInfC{$\forall x (P(x) \to P(f(x,0))), \forall x (P(x) \to P(s(x))) \vdash \forall x (P(x) \to P(f(x,l)))$}
	\RightLabel{$\forall_r$}
	\UnaryInfC{$\forall x (P(x) \to P(f(x,0))), \forall x (P(x) \to P(s(x))) \vdash \forall n \forall x (P(x) \to P(f(x,n)))$}
\end{prooftree}
\end{scriptsize}
and $\varphi_1$, $\varphi_2$, and $\varphi_3$ are proofs without induction rules.
Now we define $\pi$ as
\begin{prooftree}
	\AxiomC{$(\varphi)$}
	\noLine
	\UnaryInfC{$\forall x (P(x) \to P(s(x))) \vdash \forall n \forall x (P(x) \to P(f(x,n)))$}
	\AxiomC{$(\pi_1)$}
	\noLine
	\UnaryInfC{$S_\pi$}
	\RightLabel{$cut$}
	\BinaryInfC{$\forall x (P(x) \to P(s(x))) \vdash \forall n \forall x ((P(f(x,n)) \to P(g(x,n))) \to (P(x) \to P(g(x,n))))$}
\end{prooftree}
$\pi_1$ is a proof without induction of the form
\begin{prooftree}
	\AxiomC{$(\pi'_1)$}
	\noLine
	\UnaryInfC{$P(u) \to P(f(u,i)) \vdash (P(f(u,i)) \to P(g(u,i))) \to (P(u) \to P(g(u,i)))$}
	\RightLabel{$\forall^*_l$}
	\UnaryInfC{$\forall n \forall x (P(x) \to P(f(x,n))) \vdash (P(f(u,i)) \to P(g(u,i))) \to (P(u) \to P(g(u,i)))$}
	\RightLabel{$\forall^*_r$}
	\UnaryInfC{$\forall n \forall x (P(x) \to P(f(x,n))) \vdash \forall n \forall x ((P(f(x,n)) \to P(g(x,n))) \to (P(x) \to P(g(x,n))))$}
\end{prooftree}

Using Gentzen's method of cut-elimination, we locate the place in the proof where $\forall n$ is introduced. In $\pi_1$ $\forall n \forall x (P(x) \to P(f(x,n)))$ is obtained from $\forall x (P(x) \to P(f(x,i)))$ by $\forall_l$. In $\varphi$ we may delete the $\forall_r$ inference yielding the cut-formula and replace $l$ by $i$. But in the attempt to eliminate the formula $\forall x (P(x) \to P(f(x,i)))$ in $\varphi$ we get stuck, as we cannot cross the $ind$ rule. Note that also $ind$ cannot be eliminated as $i$ is a variable. This problem is not due to the specific form of $\varphi$ nor of $ind$. In fact, there exists no proof of $S$ in $\LK$ with only atomic cuts, even if $ind$ is used. Induction on the formula 
$$ \forall n \forall x ((P(f(x,n)) \to P(g(x,n))) \to (P(x) \to P(g(x,n))))$$
fails. To prove the end-sequent an inductive lemma is needed, i.e. something which implies $\forall n \forall x (P(x) \to P(f(x,n)))$ and cannot be eliminated.
\end{example}
There are methods for performing cut-elimination in presence of induction \cite{brotherston2011sequent,mcdowell2000cut}, however, the resulting proofs do not have the subformula property and Herbrand's theorem cannot be realized.
If induction is represented via schemata of proofs (see e.g.~\cite{DLRW13,LPW17}) schematic cut-elimination methods can be defined which allow the extraction of so-called {\em Herbrand schemata}, i.e. a generalization of Herbrand's theorem to schematic proofs. 
The underlying cut-elimination method is {\em schematic Ceres} \cite{DLRW13,LPW17,Thesis.Lolic.2020}, and in \cite{BHLRS.2008} the method was successfully applied to F\"urstenberg's proof of the infinitude of primes, which is a complicated proof using topological concepts (see~\cite{Fuerstenberg}). F\"urstenberg's proof was formalized as a first-order schema, i.e. a sequence of proofs indexed by the number of primes assumed to exist, and schematic Ceres was applied to the entire sequence. The analysis was performed in a semi-automated way, and  major parts of the analysis had to be performed by hand. Nevertheless, the analysis showed that from F\"urstenberg's proof Euclid's elementary proof could be obtained. Though a fully automated analysis of this proof is not yet within reach, this example reveals the need for the development of a formal language for analyzing proofs with induction.

An open question has been the relation of proof schemata to systems of arithmetic. In particular it was not known whether the classes of proofs specifiable in primitive recursive arithmetic and via proof schemata coincide.  In \cite{LPAR2024C} it was shown that quantifier-free proofs in primitive recursive arithmetic (quantifier-free PRA-proofs) can be simulated by the proof schema formalism. 
In this paper we generalize and extend the results from \cite{LPAR2024C} by considering also cases where quantifiers occur in PRA-proofs. We demonstrate that the proof in Example \ref{ex.PRA-proof} cannot only be formalized as a proof schema, but also its schematic Herbrand sequent can be computed.
\section{Schematic Language and PRA}

\begin{figure}[t]
    \centering

\begin{minipage}[b]{.5 \textwidth}

% Ax
\begin{prooftree} 
\def\fCenter{\mbox{\ $\Rightarrow$\ }}
\AxiomC{}
\RightLabel{$Axiom$}
\UnaryInfC{$A \vdash A$}
\end{prooftree}
\end{minipage}\begin{minipage}[b]{.5 \textwidth}
% Cut
\begin{prooftree} 
\def\fCenter{\mbox{\ $\Rightarrow$\ }}
\AxiomC{$\Gamma \vdash \Delta, F$}
\AxiomC{$\Sigma, F \vdash \Pi$}
\RightLabel{$cut$}
\BinaryInfC{$\Gamma, \Sigma \vdash \Delta, \Pi$}
\end{prooftree}
\end{minipage}

 \begin{minipage}[b]{.5 \textwidth}
% I - weak L
\begin{prooftree} 
\def\fCenter{\mbox{\ $\Rightarrow$\ }}
\AxiomC{$ \Gamma \vdash \Delta$}
\RightLabel{$w_l$}
\UnaryInfC{$  F , \Gamma \vdash \Delta$}
\end{prooftree}
\end{minipage}\begin{minipage}[b]{.5 \textwidth}
% I - weak R
\begin{prooftree} 
\def\fCenter{\mbox{\ $\Rightarrow$\ }}
\AxiomC{$\Gamma \vdash \Delta$}
\RightLabel{$w_r$}
\UnaryInfC{$ \Gamma \vdash \Delta, F$}
\end{prooftree}
\end{minipage}

\begin{minipage}[b]{.5 \textwidth}

% I - contr L
\begin{prooftree} 
\def\fCenter{\mbox{\ $\Rightarrow$\ }}
\AxiomC{$ F,F, \Gamma \vdash \Delta$}
\RightLabel{$c_l$}
\UnaryInfC{$ F, \Gamma \vdash \Delta$}
\end{prooftree}
\end{minipage}\begin{minipage}[b]{.5 \textwidth}
% I - contr R
\begin{prooftree} 
\def\fCenter{\mbox{\ $\Rightarrow$\ }}
\AxiomC{$ \Gamma \vdash \Delta, F,F$}
\RightLabel{$c_r$}
\UnaryInfC{$  \Gamma \vdash \Delta, F$}
\end{prooftree}
\end{minipage}

\begin{minipage}[b]{.5 \textwidth}

% and l
\begin{prooftree} 
\def\fCenter{\mbox{\ $\Rightarrow$\ }}
\AxiomC{$\Gamma, F, G \vdash \Delta$}
\RightLabel{$\land_l$}
\UnaryInfC{$\Gamma, F \land G \vdash \Delta$}
\end{prooftree}
\end{minipage}\begin{minipage}[b]{.5 \textwidth}
% and r
\begin{prooftree} 
\def\fCenter{\mbox{\ $\Rightarrow$\ }}
\AxiomC{$\Gamma \vdash \Delta, F$}
\AxiomC{$\Sigma \vdash \Pi, G$}
\RightLabel{$\land_r$}
\BinaryInfC{$\Gamma, \Sigma \vdash \Delta, \Pi, F \land G$}
\end{prooftree}
\end{minipage}

\begin{minipage}[b]{.5 \textwidth}
% or l
\begin{prooftree} 
\def\fCenter{\mbox{\ $\Rightarrow$\ }}

\AxiomC{$\Gamma , F \vdash \Delta$}
\AxiomC{$\Sigma , G \vdash \Pi$}
\RightLabel{$\lor_l$}
\BinaryInfC{$\Gamma, \Sigma, F \lor G \vdash \Delta, \Pi$}
\end{prooftree}
\end{minipage}\begin{minipage}[b]{.5 \textwidth}
% or r
\begin{prooftree} 
\def\fCenter{\mbox{\ $\Rightarrow$\ }}
\AxiomC{$\Gamma\vdash \Delta, F, G$}
\RightLabel{$\lor_r$}
\UnaryInfC{$\Gamma \vdash \Delta, F \lor G$}
\end{prooftree}
\end{minipage}

\begin{minipage}[b]{.5 \textwidth}
% neg L
\begin{prooftree} 
\AxiomC{$\Gamma \vdash \Delta, F$}
\RightLabel{$\neg_l$}
\UnaryInfC{$\Gamma , \neg F \vdash \Delta$}
\end{prooftree}
\end{minipage}\begin{minipage}[b]{.5 \textwidth}
% neg R
\begin{prooftree} 
\AxiomC{$ \Gamma, F \vdash \Delta$}
\RightLabel{$\neg_r$}
\UnaryInfC{$\Gamma \vdash \Delta, \neg F$}
\end{prooftree}
\end{minipage}

\begin{minipage}[b]{.5 \textwidth}

% to l
\begin{prooftree}
\def\fCenter{\mbox{\ $\Rightarrow$\ }}

\AxiomC{$\Gamma \vdash \Delta, F $}
%This is the leftmost branch

\AxiomC{$\Sigma , G \vdash \Pi  $}
%This is the rightmost branch

\RightLabel{$\to_l$}
\BinaryInfC{$ \Gamma , \Sigma, F \to G \vdash \Delta, \Pi $}

\end{prooftree}

\end{minipage}\begin{minipage}[b]{.5 \textwidth}
% to R
\begin{prooftree} 
\AxiomC{$\Gamma, F \vdash \Delta, G$}
\RightLabel{$\to_r$}
\UnaryInfC{$\Gamma  \vdash   \Delta, F \to G$}
\end{prooftree}
\end{minipage}

\begin{minipage}[b]{.5 \textwidth}
\begin{prooftree}
    \AxiomC{$\Gamma, F[t/ x] \vdash \Delta$}
    \RightLabel{$\forall_l$}
    \UnaryInfC{$\Gamma, \forall x F[x] \vdash \Delta$}
\end{prooftree}
\end{minipage}\begin{minipage}[b]{.5 \textwidth}
\begin{prooftree}

    \AxiomC{$\Gamma \vdash F[y / x], \Delta$}
    \RightLabel{$\forall_r$}
    \UnaryInfC{$\Gamma \vdash \forall F, \Delta$}
    \end{prooftree}
\end{minipage}

\begin{minipage}[b]{.5 \textwidth}
\begin{prooftree}

    \AxiomC{$\Gamma , F[y/x] \vdash \Delta$}
    \RightLabel{$\exists_l$}
    \UnaryInfC{$\Gamma, \exists F \vdash \Delta$}
    \end{prooftree}
\end{minipage}\begin{minipage}[b]{.5 \textwidth}
\begin{prooftree}

    \AxiomC{$\Gamma \vdash F[t/x], \Delta$}
    \RightLabel{$\exists_r$}
    \UnaryInfC{$\Gamma \vdash \exists F, \Delta$}
    \end{prooftree}
\end{minipage}

    \caption{The rules for \textbf{LK}.
    In the rule $Axiom$, $A$ is quantifier free. In the rules $\forall_r$ and $\exists_l$ the variable $y$, the eigenvariable, cannot occur free in the lower sequent. In the rules $\forall_l$ and $\exists_r$ $t$ is an arbitrary term. Note that we use a restricted form of \textbf{LK} as we do not allow the application of $\forall_r$ and $\exists_l$ on induction variables.
    }
    \label{fig:LK}
\end{figure}

In~\cite{Thesis.Lolic.2020} and~\cite {CLL.2021} schematic first-order languages were defined as a basis for inductive proof analysis. These languages are based on primitive recursive definitions of function symbols, predicate symbols and proofs. In this paper we focus on the language of primitive recursive arithmetic {\PRA} as defined in~\cite{Girard} where we have only the equality as predicate symbol but infinitely many function symbols. The theory contains quantifier-free axioms and an induction rule with quantifier-free induction formulas. 
%%%%

We define the respective calculus as Gentzen's \textbf{LK} (see Figure \ref{fig:LK}), extended by an equational theory, and incorporating the following induction rule:

\begin{prooftree}
    \AxiomC{$\Gamma \vdash \Delta, F(0)$}
    \AxiomC{$\Gamma, F(y) \vdash \Delta, F(y+1)$}
    \RightLabel{$ind$}
    \BinaryInfC{$\Gamma \vdash \Delta, F(n)$}
\end{prooftree}
where $y$ and $n$ are variables of sort $\omega$ (the type of the natural numbers). $y$ does not occur in $\Gamma, \Delta, F(0)$ and $F$ is quantifier-free.
Note that in \cite{Girard}, the induction variable is an arbitrary term. Our restriction to a variable of sort $\omega$ is equivalent, as it is easy to see that any arbitrary term can be simulated in the conclusion.
The calculus resulting from the combination of the rules from Figure \ref{fig:LK}, an equational theory $\mathcal{E}$ and $ind$ is denoted by {\PRA}.

%%%%%%
The language is schematic in the sense that, in addition to successor $s$ and constant $\zerob$, it contains a countably infinite set of function symbols
$\Fcal$ (of arbitrary arity) together with primitive recursive definitions of functions based on these symbols. So if $g,h$ are in $\Fcal$, $h$ is $n$ -- ary and $g$ is $n+2$ -- ary we can choose an $n+1$ -- ary symbol $f$ and add a set of two equations 
$$\Ecal(f) = \{f(\vec{x},\zerob)=  h(\vec{x}),\ f(\vec{x},s(y)) = g(\vec{x},y,f(\vec{x},y))\}.$$
The function symbols in $\Fcal$ have to be partially ordered such that in any definition $\Ecal(f)$ as defined above $h<f$ and $g<f$. With $\Ecal$ we denote the union of the sets $\Ecal(f)$ for 
$f \in \Fcal$. The minimal elements in $\Fcal$ must fulfill the condition that 
$h(\vec{x})$ and $g(\vec{x},y,z)$ can be identified with terms over the signature $\{s,\zerob\}$. As an example we may consider the definition of $+$ and $*$ via two binary function symbols $f$ and $g$:
\begin{eqnarray*}
\Ecal(f) &=& \{f(x,\zerob) = x,\ f(x,s(y)) = s(f(x,y))\},\\
\Ecal(g) &=& \{g(x,\zerob) = \zerob,\ g(x,s(y)) = f(x,g(x,y))\}.
\end{eqnarray*}  
Here $f<g$ and $f$ is minimal. The characteristic feature of {\PRA} is the fact that every ground term in this language evaluates to a {\em numeral}; numerals are elements of the form $s^n(\zerob)$ for $n \in \N$, the set of all numerals is denoted by $\Num$. 
We introduce the concept of {\em parameter} to classify a subset of the first-order variables $V$ which are supposed to be evaluated. Thereby we fix the domain of interpretation for {\PRA} to the standard model of arithmetic with domain $\N$ and $0$ for $\zerob$ and successor for $s$. Parameters will play the role of induction eigenvariables in proofs. The set of parameters is denoted by $\Pcal$. As an example, a schematic version of  Example~\ref{ex.PRA-proof} would consider $k$ as a parameter and the other variables as ''ordinary'' first-order variables. For proof schemata to be defined in Section~\ref{sec.proofschema} the parameters will be crucial in defining the semantics of schemata. 
\begin{definition}[parameter assignment]\label{def.par-assign}
A {\em parameter assignment} $\sigma$ is a mapping $\Pcal \to \Num$ with the following extensions to terms and formulas:
\begin{itemize}
\item $\sigma(x) = x$ for $x \in V \setminus \Pcal$,
\item $\sigma(\zerob) = \zerob,\ \sigma(s(t)) = s(\sigma(t))$,
\item For an $n$ -- ary $f \in \Fcal$ we have $\sigma(f(t_1,\ldots,t_n)) = f(\sigma(t_1),\ldots,\sigma(t_n))$,
\item $\sigma(\neg A) = \neg \sigma(A)$,
\item $\sigma(A \circ B)  = \sigma(A) \circ \sigma(B)$ for $\circ \in \{\land,\lor,\impl\}$,
\item $\sigma(Qx.A) = Qx.\sigma(A)$ if $x \in V \setminus \Pcal$ and $Q \in \{\forall,\exists\}$.
\end{itemize}
The set of all parameter assignments is denoted by $\Scal$.
\end{definition}
As already mentioned every term containing parameters can be evaluated under parameter assignments; 
terms containing only parameters and no variables in $V \setminus \Pcal$ evaluate to numerals, others are merely {\em partially evaluated}.  
\begin{definition}[evaluation]\label{def.par-evaluation}
Let $f$ be an $(n+1)$ -- ary function symbol in $\Fcal$ with $n \in \Pcal$ and 
$$\Ecal(f) = \{f(\vec{x},\zerob)=  h(\vec{x}),\ f(\vec{x},s(n)) = g(\vec{x},y,f(\vec{x},n))\} \mbox{ and}$$
$t_1,\ldots,t_l$ be terms, $\vec{t} = (t_1,\ldots,t_l)$ and $k \in \Num$. Then, 
\begin{itemize}
\item $f(\vec{t},\zerob)\Eval = h(\vec{t})\Eval$,
\item $f(\vec{t},s(\bar{k}))\Eval = g(\vec{t},\bar{k},f(\vec{t},\bar{k})\Eval))\Eval$, 
\item $r\Eval = r$ if $r$ is a term not containing a symbol in $\Fcal$.
\end{itemize}
The evaluation operator has to be applied recursively to all subterms of a term containing symbols in $\Fcal$. 
Assume that $\sigma$ is a parameter assignment and $t$ is a term with $\sigma(t)\Eval = \bar{k}$ and $r \in \Fcal$ then 
$\sigma(r(\vec{s},t))\Eval =r(\sigma(\vec{s}))\Eval,\bar{k})\Eval$.
\end{definition}
The evaluation under $\sigma$ can be extended to formulas homomorphically in an obvious way.
\begin{example}\label{ex.evaluation}
Let $x,n \in V$, $x \in V\setminus \Pcal$ and $n \in \Pcal$; assume that $\sigma(n) = \bar{1}$. Let $f \in \Fcal$ and 
$$\Ecal(f) = \{f(x,\zerob) = x,\ f(x,s(n)) = s(f(x,n))\}.$$
Then $\sigma(f(x,s(n)))\Eval = f(\sigma(x)\Eval,\sigma(s(n))\Eval) = f(x,\bar{2})\Eval$ by $\sigma(x)=x$ and 
\begin{itemize}
\item $\sigma(s(n))\Eval = \sigma(s(n)) = s(\sigma(n)) = s(\bar{1}) =\bar{2}$. 
\end{itemize}
Furthermore 
$$f(x,\bar{2})\Eval = s(f(x,\bar{1}))\Eval = s(s(f(x,\zerob)))\Eval  = s(s(x))$$
and so $\sigma(f(x,s(n))\Eval = s(s(x))$. We see that by distinguishing parameters and variables we obtain just a partial evaluation, i.e. terms need not evaluate to numerals but just to other terms containing variables.  \\[1ex]
Let $A = \forall x.f(x,n) = f(n,x)$ and $\sigma(n)$ as above. Then 
$$\sigma(A)\Eval = \forall x.s(x) = f(\bar{1},x).$$
We also see that $\{\sigma(A) \mid \sigma \in \Scal\}$  defines the infinite set of formulas $\{\forall x.s^n(x) = f(\bar{n},x) \mid n \in \N\}$. So we can consider the formula $A$ as a {\em formula schema} describing an infinite sequence of formulas
\end{example}

\section{Proof Schema}\label{sec.proofschema}

The general idea of a proof schema is to represent a proof containing induction by a finite description of an infinite sequence of proofs without induction inferences: Assume a proof $\varphi$ of the end-sequent $\vdash \forall x A(x)$ that uses an induction inference. Instead of considering the proof $\varphi$, we instead consider the infinite sequence of proofs $\varphi_0$, $\varphi_1$, $\varphi_2$, $\ldots$ of end-sequents $\vdash A(\bar{0})$, $\vdash A(\bar{1})$, $\vdash A(\bar{2})$ $\ldots$
The task is to find a finite description of this infinite sequence of proofs, the proof schema. A proof schema always represents a parameterized sequence, and an evaluation under a parameter assignment $\bar{n}$ results in the proof $\varphi_n$ of $\vdash A(\bar{n})$.
The underlying problem, that initially lead to the development of proof schemata, is to be able to analyze inductive proofs and extract their Herbrand sequents. Indeed, each of the proofs $\varphi_0$, $\varphi_1$, $\varphi_2$, $\ldots$ is a simple {\LK}-proof without induction inferences, and thus enjoys cut-elimination resulting in an analytic proof. 
The concept of proof schema was initially introduced in \cite{DLRW13,LPW17} to address schemata involving a single parameter. Later, it was expanded to accommodate an arbitrary number of induction parameters \cite{Thesis.Lolic.2020}.

Formally, proof schemata are constructed using proofs in an extension of {\LK} by an equational theory. 
First, let us define the concept of schematic sequents (the definitions below are from \cite{Thesis.Lolic.2020}).
\begin{definition}[schematic sequents]
A schematic sequent is a sequent of the form $F_1 , \ldots , F_\alpha \vdash G_1 , \ldots , G_\beta$ where the $F_i$ and $G_j$ for $1 \leq i \leq \alpha$ and $1 \leq j \leq \beta$ are formula schemata.
Let $S \colon F_1 , \ldots , F_\alpha \vdash G_1 , \ldots , G_\beta$ be a schematic sequent and $\sigma$ a parameter assignment. Then the evaluation of $S$ under $\sigma$ is $\sigma(S)\Eval \colon \sigma(F_1)\Eval , \ldots , \sigma(F_\alpha)\Eval \vdash \sigma(G_1)\Eval , \ldots , \sigma(G_\beta)\Eval$. 
\end{definition}
\begin{definition} \label{def.LKE}
Let $\mathcal{E}$ be an equational theory. We extend the calculus \LK\ by the
$\mathcal{E}$ inference rule
$\dfrac{S(t)}{S(t')}\mathcal{E}$
where the term or input term schema $t$ in the schematic sequent $S$ is replaced by a term or input term schema $t'$ where $t=t'$ is an instance of an equation in $\Ecal$.
\end{definition}
The definitions below will use the schematic standard axiom set $\Acal_s$.
\begin{definition}[schematic standard axiom set] 
Let $\Acal_s$ be the smallest set of schematic sequents that is closed under substitution containing all sequents of the form $A \vdash A$ for arbitrary atomic formula schemata $A$. Then $\Acal_s$ is called the schematic standard axiom set. 
\end{definition}
Schematic derivations can be understood as parameterized sequences of $\LK$-derivations where new kinds of axioms in the form of labeled sequents are included. These labeled sequents serve the purpose to establish recursive call structures in the proof.   
For constructing schematic derivations we introduce a countably infinite set $\Delta$ of {\em proof symbols} which are used to label the individual proofs of a proof schema. A proof schema uses only a finite set of proof symbols $\Delta^*\subset \Delta$. We assign an arity $A(\delta)$ to every $\delta \in \Delta^*$, $A(\delta)$ is the arity of the input parameters for the proof labeled by $\delta$. Also, we need a concept of proof labels which serve the purpose to relate some leafs of the proof tree to recursive calls. 
\begin{definition}[proof label]\label{def.prooflabel}
Let $\delta \in \Delta$ and $\vartheta$ be a parameter substitution. Then the pair $(\delta,\vartheta)$ is called a proof label.
\end{definition}
\begin{definition}[labeled sequents and derivations]
Let $S$ be a schematic sequent and $(\delta,\vartheta)$ a proof label, then $(\delta,\vartheta)\colon S$ is a {\em labeled sequent}. A {\em labeled} derivation is a derivation $\pi$ where all leaves are labeled.
\end{definition}  
In the definition below we will define a proof schema over a base-case proof (for parameter $0$) and a step-case proof (for parameter $m+1$), where initial sequents are either axioms, or end-sequents from previously defined base- or step-case proofs. 
In general, the step-case proof for some proof symbol $\delta$ uses as initial sequent its own end-sequent, but under a parameter assignment $m$. Evaluating a schematic derivation means that initial sequents, which are no axioms, have to be replaced by their derivations.  
\begin{definition}[parameter replacement]
Let $\vec{m},\vec{n}$ be tuples of parameters. A parameter replacement on $\vec{n}$ with respect to $\vec{m}$ is a replacement substituting every parameter $p$ in $\vec{n}$ by a term $t_p$, where the parameters of $t_p \in T^\omega$ are in $\vec{m}$.
\end{definition}
\begin{definition}[schematic deduction and proof schema]\label{def.proofschema} 
Let $\Dcal$ be the tuple $(\delta_0, \Delta^*, \Pi)$. $\Dcal$ is called a schematic deduction from a finite set of schematic sequents $\Scal$ if the following conditions are fulfilled:
\begin{itemize}
	\item $\Delta^*$ is a finite subset of $\Delta$.
	\item $\delta_0 \in \Delta^*$, and $\delta_0 > \delta'$ for all $\delta' \in \Delta^*$ such that $\delta' \neq \delta_0$. $\delta_0$ is called the main symbol.
	\item To every $\delta \in \Delta^*$ we assign a parameter tuple $\vec{n}_\delta$ of pairwise different parameters (called the passive parameters), and a parameter $m_\delta$ (called the active parameter).
	\item $\Pi$ is a set of pairs $\{(\Pi(\delta, \vec{n}_\delta, m_\delta), S(\delta, \vec{n}_\delta, m_\delta)\}$, where $S(\delta, \vec{n}_\delta, m_\delta)$ is a schematic sequent, and 
	$$\Pi(\delta, \vec{n}_\delta, m_\delta) = \{(\delta,\vec{n}_\delta, m_\delta) \to \rho(\delta,\vec{n}_\delta, m_\delta)\},$$
	where $\rho(\delta,\vec{n}_\delta, 0) = \rho_0(\delta,\vec{n}_\delta)$, and $\rho(\delta,\vec{n}_\delta, s(m_\delta)) = \rho_1(\delta,\vec{n}_\delta, m_\delta)$, 
	and there exists a (possibly empty) finite set of schematic sequents $\Ccal(\delta)$ such that 
\begin{enumerate}
	\item $\rho_0(\delta,\vec{n}_\delta)$ is a deduction of $S(\delta, \vec{n}_\delta, 0)$ from $\Scal \union \Ccal(\delta)$,
	\item $\rho_1(\delta,\vec{n}_\delta, m_\delta)$ is a deduction of $S(\delta, \vec{n}_\delta, m_\delta+1)$ from $\{(\delta, \Psi) \colon S(\delta, \vec{n}_\delta, m_\delta)\} \union \Scal \union \Ccal(\delta)$, where $(\delta, \Psi)$ is a label, and $\Psi$ the empty parameter replacement,
\item for all $S' \in \Ccal(\delta)$, $S' = (\delta', \Psi) \colon S(\delta', \vec{n}_{\delta'}, m_{\delta'})\Psi$ where $(\delta', \Psi)$ is a label, $\delta' \in \Delta^*$ with $\delta > \delta'$ and $\Psi$ is a parameter replacement on $(\vec{n}_{\delta'}$, $m_{\delta'})$ w.r.t. $(\vec{n}_\delta, m_\delta)$ such that the conditions $1.$ and $2.$ hold for $\delta'$ . 
\end{enumerate}
If $\Scal = \Acal_s$ we call $\Dcal$ a {\em proof schema} of $S(\delta_0, \vec{n}_{\delta_0}, m_{\delta_0})$.
\end{itemize}
\end{definition}
In the example below we will formalize the PRA-proof from Example \ref{ex.PRA-proof} as a proof schema.
\begin{example}\label{ex.proofschema}
In this example we will construct a proof schema of the end-sequent
$$\forall x (P(x) \to P(s(x))) \vdash \forall x (P(f(x,n)) \to P(g(x,n))) \to (P(x) \to P(g(x,n))) $$ 
where 
$$f(x,0) = x, \quad f(x,s(n)) = s(f(x,n)).$$
Note that this end-sequent is equivalent over the standard model to this in Example \ref{ex.PRA-proof} for which no Herbrand sequent could be obtained. The only difference is that the $\forall n$ in the end-sequent of Example \ref{ex.PRA-proof} disappears in the proof schema formalism, as $n$ will be the schema's inductive parameter and we do not allow strong quantification of inductive parameters. Note further that in this example we consider only one parameter, but in general, we allow arbitrarily many.
To start, let us first define a proof schema 
$$\Dcal_1 = \{(\delta, \rho(\delta,0),\rho(\delta,n+1))\},$$
with end-sequent
$$S(\delta) \colon \forall x (P(x) \to P(s(x))) \vdash \forall x (P(x) \to P(f(x,n))).$$
We define $\rho(\delta,0)$ as follows:
\begin{prooftree}
	\AxiomC{$P(f(a,0)) \vdash P(f(a,0))$}
	\RightLabel{$\Ecal$}
	\UnaryInfC{$P(a) \vdash P(f(a,0))$}
	\RightLabel{$\to_r$}
	\UnaryInfC{$\vdash P(a) \to P(f(a,0))$}
	\RightLabel{$\forall_r$}
	\UnaryInfC{$\vdash \forall x (P(x) \to P(f(x,0)))$}
	\RightLabel{$w_l$}
	\UnaryInfC{$\forall x (P(x) \to P(s(x))) \vdash \forall x (P(x) \to P(f(x,0)))$}
\end{prooftree}
$\rho(\delta,n+1)$ is defined as follows:
\begin{prooftree}
	\AxiomC{$(\delta, \emptyset) \colon S(\delta)$}
	\AxiomC{$(1)$}
	\RightLabel{$cut, c_l$}
	\BinaryInfC{$\forall x (P(x) \to P(s(x))) \vdash \forall x (P(x) \to P(f(x,n+1)))$}
\end{prooftree}
where $(1)$ is
\begin{prooftree}
	\AxiomC{$P(a) \vdash P(a)$}
	\AxiomC{$P(f(a,n)) \vdash P(f(a,n))$}
	\AxiomC{$P(f(a,n+1)) \vdash P(f(a,n+1))$}
	\RightLabel{$\Ecal$}
	\UnaryInfC{$P(s(f(a,n))) \vdash P(f(a,n+1))$}
	\RightLabel{$\to_l$}
	\BinaryInfC{$P(f(a,n)), P(f(a,n)) \to P(s(f(a,n))) \vdash P(f(a,n+1))$}
	\RightLabel{$\forall_l$}
	\UnaryInfC{$P(f(a,n)), \forall x (P(x) \to P(s(x))) \vdash P(f(a,n+1))$}
	\RightLabel{$\to_l$}
	\BinaryInfC{$P(a), P(a) \to P(f(a,n)), \forall x (P(x) \to P(s(x))) \vdash P(f(a,n+1))$}
	\RightLabel{$\to_r$}
	\UnaryInfC{$P(a) \to P(f(a,n)), \forall x (P(x) \to P(s(x))) \vdash P(a) \to  P(f(a,n+1))$}
	\RightLabel{$\forall_l$}
	\UnaryInfC{$\forall x (P(x) \to P(f(x,n))), \forall x (P(x) \to P(s(x))) \vdash P(a) \to  P(f(a,n+1))$}
	\RightLabel{$\forall_r$}
	\UnaryInfC{$\forall x (P(x) \to P(f(x,n))), \forall x (P(x) \to P(s(x))) \vdash \forall x (P(x) \to  P(f(x,n+1)))$}
\end{prooftree}
Now we construct the proof schema 
$$\Dcal = \{(\delta', \rho(\delta',n),\rho(\delta',n+1))\} \cup \Dcal_1,$$
with the end-sequent
$$S(\delta') \colon \forall x (P(x) \to P(s(x))) \vdash \forall x (P(f(x,n)) \to P(g(x,n))) \to (P(x) \to P(g(x,n))),$$
and $\delta' > \delta$. There is no internal recursion in $\delta'$ needed, hence we only define $\rho(\delta',n)$ as follows:
\begin{prooftree}
	\AxiomC{$(\delta, \emptyset) \colon S(\delta)$}
	\AxiomC{$(2)$}
	\RightLabel{$cut$}
	\BinaryInfC{$\forall x (P(x) \to P(s(x))) \vdash \forall x (P(f(x,n)) \to P(g(x,n))) \to (P(x) \to P(g(x,n)))$}
\end{prooftree}
where $(2)$ is
\begin{prooftree}
	\AxiomC{$P(c) \vdash P(c)$}
	\AxiomC{$P(f(c,n)) \vdash P(f(c,n))$}
	\AxiomC{$P(g(c,n)) \vdash P(g(c,n))$}
	\RightLabel{$\to_l$}
	\BinaryInfC{$P(f(c,n)) \to P(g(c,n)), P(f(c,n)) \vdash P(g(c,n))$}
	\RightLabel{$\to_l$}
	\BinaryInfC{$P(c), P(f(c,n)) \to P(g(c,n)), P(c) \to P(f(c,n)) \vdash P(g(c,n))$}
	\RightLabel{$\to_r$}
	\UnaryInfC{$P(f(c,n)) \to P(g(c,n)), P(c) \to P(f(c,n)) \vdash P(c) \to P(g(c,n))$}
	\RightLabel{$\to_r$}
	\UnaryInfC{$P(c) \to P(f(c,n)) \vdash (P(f(c,n)) \to P(g(c,n))) \to (P(c) \to P(g(c,n)))$}
	\RightLabel{$\forall_l$}
	\UnaryInfC{$\forall x (P(x) \to P(f(x,n))) \vdash (P(f(c,n)) \to P(g(c,n))) \to (P(c) \to P(g(c,n)))$}
	\RightLabel{$\forall_r$}
	\UnaryInfC{$\forall x (P(x) \to P(f(x,n))) \vdash \forall x (P(f(x,n)) \to P(g(x,n))) \to (P(x) \to P(g(x,n)))$}
\end{prooftree}
\end{example}
Proof schemata define infinite sequences of proofs, and can be evaluated under parameter assignments. 
\begin{definition}[evaluation of proof schema]
Let $\Dcal = (\delta_0, \Delta^*, \Pi)$ be a proof schema, and $\sigma$ a parameter assignment. In defining the evaluation of the proof schema, $\Dcal\sigma\Eval$, we proceed by double induction on the ordering of proof symbols and the assignments $\sigma$.
\begin{itemize}
\item Let $\delta_i$ be a minimal element in $\Delta^*$. 
\begin{enumerate}
\item $\sigma(m_{\delta_i}) = 0$.

Then, by definition of a proof schema, $\rho_0(\delta_i,\vec{n}_{\delta_i})$ is a proof with $\LK$-inferences and inferences for defined symbols that contain schematic sequents. Let $S_1, \ldots S_n$ be all the schematic sequents in $\rho_0(\delta_i,\vec{n}_{\delta_i})$. Then the evaluation of $\rho_0(\delta_i,\vec{n}_{\delta_i})$ under $\sigma$ is denoted by $\rho_0(\delta_i,\vec{n}_{\delta_i})\Eval$ and obtained by replacing all $S_1, \ldots, S_n$ in $\rho_0(\delta_i,\vec{n}_{\delta_i})$ by $\sigma(S_1)\Eval, \ldots, \sigma(S_n)\Eval$ and omitting the inferences for defined symbols. 
\item $\sigma(m_{\delta_i}) = \alpha >0$.

Evaluate all schematic sequents except the leaves $(\delta_i, \emptyset) \colon S(\delta_i, \vec{n}_{\delta_i},m_{\delta_i})$ under $\sigma$.
Let $\sigma[m_{\delta_i} /\alpha-1]$ be defined as $\sigma[m_{\delta_i}/\alpha-1](p) = \sigma(p)$ for all $p \neq m_{\delta_i}$ and $\sigma[m_{\delta_i}/\alpha-1](m_{\delta_i}) = \alpha-1$.
Then we replace the labeled sequent $(\delta_i, \emptyset) \colon S(\delta_i, \vec{n}_{\delta_i},m_{\delta_i})$ by the proofs $\rho_0(\delta_i,\vec{n}_{\delta_i})\sigma[m_{\delta_i} /\alpha - 1]\Eval$ if $\alpha - 1 = 0$ and by $\rho_1(\delta_i,\vec{n}_{\delta_i},m_{\delta_i})\sigma[m_{\delta_i} /\alpha - 1]\Eval$ if $\alpha - 1 > 0$.
The result is an $\LK$-proof.
\end{enumerate}
\item $\delta_i \in \Delta^*$ is not minimal.
\begin{enumerate}
\item $\sigma(m_{\delta_i})=0$.

Evaluate all schematic sequents except the labeled sequents of the form $(\delta', \Psi) \colon \\ S(\delta',\vec{n}_{\delta'}, m_{\delta'})\Psi$ for $\delta_i > \delta'$ and the corresponding parameter replacement $\Psi$ under $\sigma$. 
Then replace the labeled sequent $(\delta', \Psi) \colon S(\delta',\vec{n}_{\delta'},m_{\delta'})\Psi$ by the proof $\rho_0(\delta',\vec{n}_{\delta'})\Psi\sigma\Eval$ if $\sigma(m_{\delta'})=0$ and by the proof $\rho_1(\delta',\vec{n}_{\delta'},m_{\delta'})\Psi\sigma\Eval$ otherwise.
\item $\sigma(m_{\delta_i}) = \alpha >0$.

As above, except for the labeled sequents $(\delta_i, \emptyset) \colon S(\delta_i,\vec{n}_{\delta_i},m_{\delta_i})$ which are replaced by the proof $\rho_0(\delta_i,\vec{n}_{\delta_i})\sigma[m_{\delta_i}/\alpha-1]$ if $\alpha-1 = 0$ and by the proof $\rho_1(\delta_i,\vec{n}_{\delta_i},m_{\delta_i})\sigma[m_{\delta_i}/\alpha-1]$ otherwise.
\end{enumerate}
\end{itemize}
$\Dcal \sigma \Eval$ is defined as $\rho_0(\delta_0,\vec{n}_{\delta_0})\sigma\Eval$ for the $<$-maximal symbol $\delta_0$ if $\sigma(m_{\delta_0})=0$, and by $\rho_1(\delta_0,\vec{n}_{\delta_0},m_{\delta_0})\sigma\Eval$ if $\sigma(m_{\delta_0})>0$.
\end{definition}
We will illustrate the evaluation of a proof schema for a concrete parameter in the example below.
\begin{example}
Let $\Dcal$ be the proof schema as defined in Example \ref{ex.proofschema}, and $\sigma(n) = 0$. Then $\Dcal\sigma\Eval$ is obtained by considering the derivation $\rho(\delta',0)$, where the proof call to $(\delta, \emptyset) \colon S(\delta) \{n \leftarrow 0\}$ is replaced by the derivation $\rho(\delta,0)$, hence we obtain
\begin{prooftree}
	\AxiomC{$P(f(a,0)) \vdash P(f(a,0))$}
	\RightLabel{$\Ecal$}
	\UnaryInfC{$P(a) \vdash P(f(a,0))$}
	\RightLabel{$\to_r$}
	\UnaryInfC{$\vdash P(a) \to P(f(a,0))$}
	\RightLabel{$\forall_r$}
	\UnaryInfC{$\vdash \forall x (P(x) \to P(f(x,0)))$}
	\RightLabel{$w_l$}
	\UnaryInfC{$\forall x (P(x) \to P(s(x))) \vdash \forall x (P(x) \to P(f(x,0)))$}
	\AxiomC{$(1)$}
	\RightLabel{$cut$}
	\BinaryInfC{$\forall x (P(x) \to P(s(x))) \vdash \forall x (P(f(x,0)) \to P(g(x,0))) \to (P(x) \to P(g(x,0)))$}
\end{prooftree}
where $(1)$ is
\begin{prooftree}
	\AxiomC{$P(c) \vdash P(c)$}
	\AxiomC{$P(f(c,0)) \vdash P(f(c,0))$}
	\AxiomC{$P(g(c,0)) \vdash P(g(c,0))$}
	\RightLabel{$\to_l$}
	\BinaryInfC{$P(f(c,0)) \to P(g(c,0)), P(f(c,0)) \vdash P(g(c,0))$}
	\RightLabel{$\to_l$}
	\BinaryInfC{$P(c), P(f(c,0)) \to P(g(c,0)), P(c) \to P(f(c,0)) \vdash P(g(c,0))$}
	\RightLabel{$\to_r$}
	\UnaryInfC{$P(f(c,0)) \to P(g(c,0)), P(c) \to P(f(c,0)) \vdash P(c) \to P(g(c,0))$}
	\RightLabel{$\to_r$}
	\UnaryInfC{$P(c) \to P(f(c,0)) \vdash (P(f(c,0)) \to P(g(c,0))) \to (P(c) \to P(g(c,0)))$}
	\RightLabel{$\forall_l$}
	\UnaryInfC{$\forall x (P(x) \to P(f(x,0))) \vdash (P(f(c,0)) \to P(g(c,0))) \to (P(c) \to P(g(c,0)))$}
	\RightLabel{$\forall_r$}
	\UnaryInfC{$\forall x (P(x) \to P(f(x,0))) \vdash \forall x (P(f(x,0)) \to P(g(x,0))) \to (P(x) \to P(g(x,0)))$}
\end{prooftree}
Now let us consider $\sigma(n) = 1$. Then $\Dcal \sigma \Eval$ is obtained by considering $\rho(\Delta',1)$ and replacing the proof link to $(\delta, \emptyset) \colon S(\delta)$ by the derivation $\rho(\delta,1)$. Note that in $\rho(\delta,1)$ (which in our formalism is denoted by $\rho(\delta,0+1)$) we have to replace the self-referencing proof link $(\delta, \emptyset) \colon S(\delta)$ by the derivation $\rho(\delta,0)$. Therefore, we obtain
\begin{prooftree}
	\AxiomC{$P(f(a,0)) \vdash P(f(a,0))$}
	\RightLabel{$\Ecal$}
	\UnaryInfC{$P(a) \vdash P(f(a,0))$}
	\RightLabel{$\to_r$}
	\UnaryInfC{$\vdash P(a) \to P(f(a,0))$}
	\RightLabel{$\forall_r$}
	\UnaryInfC{$\vdash \forall x (P(x) \to P(f(x,0)))$}
	\RightLabel{$w_l$}
	\UnaryInfC{$\forall x (P(x) \to P(s(x))) \vdash \forall x (P(x) \to P(f(x,0)))$}
	\AxiomC{$(1)$}
	\RightLabel{$cut, c_l$}
	\BinaryInfC{$\forall x (P(x) \to P(s(x))) \vdash \forall x (P(x) \to P(f(x,1)))$}
	\AxiomC{$(2)$}
	\RightLabel{$cut$}
	\BinaryInfC{$\forall x (P(x) \to P(s(x))) \vdash \forall x (P(f(x,1)) \to P(g(x,1))) \to (P(x) \to P(g(x,1)))$}
\end{prooftree}
where $(1)$ is
\begin{prooftree}
	\AxiomC{$P(a) \vdash P(a)$}
	\AxiomC{$P(f(a,0)) \vdash P(f(a,0))$}
	\AxiomC{$P(f(a,1)) \vdash P(f(a,1))$}
	\RightLabel{$\Ecal$}
	\UnaryInfC{$P(s(f(a,0))) \vdash P(f(a,1))$}
	\RightLabel{$\to_l$}
	\BinaryInfC{$P(f(a,0)), P(f(a,0)) \to P(s(f(a,0))) \vdash P(f(a,1))$}
	\RightLabel{$\forall_l$}
	\UnaryInfC{$P(f(a,0)), \forall x (P(x) \to P(s(x))) \vdash P(f(a,1))$}
	\RightLabel{$\to_l$}
	\BinaryInfC{$P(a), P(a) \to P(f(a,0)), \forall x (P(x) \to P(s(x))) \vdash P(f(a,1))$}
	\RightLabel{$\to_r$}
	\UnaryInfC{$P(a) \to P(f(a,0)), \forall x (P(x) \to P(s(x))) \vdash P(a) \to  P(f(a,1))$}
	\RightLabel{$\forall_l$}
	\UnaryInfC{$\forall x (P(x) \to P(f(x,0))), \forall x (P(x) \to P(s(x))) \vdash P(a) \to  P(f(a,1))$}
	\RightLabel{$\forall_r$}
	\UnaryInfC{$\forall x (P(x) \to P(f(x,0))), \forall x (P(x) \to P(s(x))) \vdash \forall x (P(x) \to  P(f(x,1)))$}
\end{prooftree}
and $(2)$ is
\begin{prooftree}
	\AxiomC{$P(c) \vdash P(c)$}
	\AxiomC{$P(f(c,1)) \vdash P(f(c,1))$}
	\AxiomC{$P(g(c,1)) \vdash P(g(c,1))$}
	\RightLabel{$\to_l$}
	\BinaryInfC{$P(f(c,1)) \to P(g(c,1)), P(f(c,1)) \vdash P(g(c,1))$}
	\RightLabel{$\to_l$}
	\BinaryInfC{$P(c), P(f(c,1)) \to P(g(c,1)), P(c) \to P(f(c,1)) \vdash P(g(c,1))$}
	\RightLabel{$\to_r$}
	\UnaryInfC{$P(f(c,1)) \to P(g(c,1)), P(c) \to P(f(c,1)) \vdash P(c) \to P(g(c,1))$}
	\RightLabel{$\to_r$}
	\UnaryInfC{$P(c) \to P(f(c,1)) \vdash (P(f(c,1)) \to P(g(c,1))) \to (P(c) \to P(g(c,1)))$}
	\RightLabel{$\forall_l$}
	\UnaryInfC{$\forall x (P(x) \to P(f(x,1))) \vdash (P(f(c,1)) \to P(g(c,1))) \to (P(c) \to P(g(c,1)))$}
	\RightLabel{$\forall_r$}
	\UnaryInfC{$\forall x (P(x) \to P(f(x,1))) \vdash \forall x (P(f(x,1)) \to P(g(x,1))) \to (P(x) \to P(g(x,1)))$}
\end{prooftree}

\end{example}

In \cite{Thesis.Lolic.2020} the Skolemized version of the proof schema in Example \ref{ex.proofschema} was analyzed and a so-called Herbrand schema was obtained. A Herbrand schema can be understood as a parameterized sequence of Herbrand instances for the quantified formulas occurring in the end-sequent. The reason for Skolemization lies in the schematic cut-elimination method. In this work we will not focus on this cut-elimination method or the proof analysis of schematic proofs, but only explain the general idea. 

We start with a proof schema of a Skolemized end-sequent. This sequent has to be Skolemized as the proof transformation steps in the method would not be sound if eigenvariables were present.
The proof schema is then split in two parts: In the first part we start with the formulas in the initial sequents that are ancestors of formulas in the end-sequent, and apply only the rules that operate on these formulas. By construction the thus obtained proof schema is cut-free, we call it the projection schema. In the second part, only the formulas in initial sequents that are ancestors of cut-formulas and the rules operating on these cut-ancestors are considered. By construction, as we only have formulas that are ancestors of cuts and are therefore cut out, we end up with the empty-sequent. Therefore, the initial sequents we started with are unsatisfiable and can be represented as an unsatisfiable schematic formula, which is called the characteristic formula schema. The characteristic formula schema can be refuted with a schematic refutation calculus defined in \cite{CLL.2021}. In \cite{LPAR2024LeitschLolic} this calculus was improved and it was shown that the Herband instances of the refutation schema can be extracted in the form of Herbrand schemata. It was shown in \cite{Thesis.Lolic.2020} that the Herbrand schema from the refutation of the characteristic formula can be combined with the Herbrand schema from the projection schema to obtain the Herbrand schema of the original proof schema.

For the Skolemized version of Example \ref{ex.proofschema}, we obtain (after some simplifications) for the formula $\forall x (P(x) \to P(s(x)))$ the schematic Herbrand instances 
$$\{x \leftarrow f(c,n)\},$$
where $c$ just denotes a Skolem constant introduced when Skolemizing the end-sequent. Therefore, for some numeral $\alpha$ the instances are
$$ c, s(c), s(s(c)), \ldots, s^{\alpha-1}(c).$$
\section{Simulation of Primitive Recursive Arithmetic Through Proof Schemata}

In this section, we will analyze the expressivity of proof schemata by showing that proof schemata simulate a restricted version of primitive recursive arithmetic, as defined in \cite{Girard}. In \cite{LPW17}, it was demonstrated that proof schemata are equivalent to a specific fragment of arithmetic known as \emph{k-simple induction}. This variant restricts the introduction of new eigenvariables through induction. Recently, we extended the simulation to capture proof schemata that allow an arbitrary number of parameters \cite{LPAR2024C}. In the latter work, we used a definition of primitive recursive arithmetic, as defined in \cite{Takeuti}, which does not admit quantifier introduction. In this paper, we allow quantifier introduction with the restriction that we do not allow strong quantification of induction variables. 
However, we will later demonstrate that proof schemata can indeed capture certain proofs involving strongly quantified induction variables.

% Following \cite{Girard}, we define the respective calculus as Gentzen's \textbf{LK} (see Figure \ref{fig:LK}), extended by an equational theory as outlined in Definition \ref{def.LKE}, and incorporating the following induction rule:

% \begin{prooftree}
%     \AxiomC{$\Gamma \vdash \Delta, F(0)$}
%     \AxiomC{$\Gamma, F(y) \vdash \Delta, F(y+1)$}
%     \RightLabel{$ind$}
%     \BinaryInfC{$\Gamma \vdash \Delta, F(n)$}
% \end{prooftree}
% where $y$ is a variable of sort $\omega$, $n$ is a variable of sort $\omega$, and $y$ does not occur in $\Gamma, \Delta, F(0)$.

% Note that in \cite{Girard}, the induction variable is an arbitrary term. Our restriction to a variable of sort $\omega$ is equivalent, as it is easy to see that any arbitrary term can be simulated in the conclusion.
% Further, every sequent $S: \Gamma \vdash \Delta$ corresponds to an equivalent formula $\mathcal{F}(S) \coloneq \bigvee \neg \Gamma \cup \Delta$.
% The calculus resulting from combining the rules from Figure \ref{fig:LK}, $\mathcal{E}$ and $ind$ is denoted by \textbf{PRA}. 

We will now show the translation from proof schemata to \textbf{PRA} and back. To do this, we note that every sequent $S: \Gamma \vdash \Delta$ corresponds to an equivalent formula $\mathcal{F}(S) \coloneq \bigvee \neg \Gamma \cup \Delta$.

\begin{lemma}
 Let $\mathcal{D}$ be a proof schema with end-sequent $S$. Then there exists a \textbf{PRA} proof of $S$.
\end{lemma}
\begin{proof}
  Let $\mathcal{D} = \{ ( \delta_i, \rho(\delta_i, \vec{n}_i,0), \rho(\delta_i, \vec{n}_i , m_i +1)  ) \mid i \in \{ 1, ..., \alpha \} \}$ with $S(\delta_i) = S_i$ and if $i < j$ then $\delta_i > \delta_j$. Hence, $S(\delta_1) = S$.
  
  We construct inductively \textbf{PRA} proofs of $\mathcal{F}(S_\gamma)$, starting with $\gamma = \alpha$. 
  Assume we constructed \textbf{PRA} proofs $\xi_{\gamma +1},..., \xi_\alpha$ of $\mathcal{F(}S_{\gamma +1}), ... , \mathcal{F}(S_\alpha)$ respectively. Our aim is to construct a \textbf{PRA} proof of $\mathcal{F}(S_\gamma)$.

  % Base Case
  In $\rho(\delta_\gamma, \vec{n}_\gamma, 0)$ replace any proof call of the form $(\delta_j,\Psi): S(\delta_j)\Psi$ by $\xi_j \Psi$ to obtain proof $\xi_\gamma^B$. Applications of $\neg_r^*$ will then yield a proof of $\vdash \mathcal{F}(S_\gamma) \{ m_\gamma \leftarrow 0 \}$.

  % Step case

In $\rho(\delta_\gamma, \vec{n}_\gamma, m_{\gamma +1})$ consider any branches that lead into a self-referencing proof call of the form $(\delta_\gamma,\emptyset): S(\delta_\gamma)$. Note that any introduction of a quantifier in these branches is cut and will not be part of the end-sequent of this derivation. Otherwise, the proof call conditions are violated.
Further note that $\rho(\delta_\gamma, \vec{n}_\gamma, m_{\gamma +1})$ is an \textbf{LK} derivation and therefore admits cut-elimination. Using this, we eliminate the introduction of all strong quantifiers in these branches and replace any self-referencing proof call of the form $(\delta_\gamma,\emptyset): S(\delta_\gamma)$ by axiom $\mathcal{F}(S(\delta_\gamma)) \vdash \mathcal{F}(S(\delta_\gamma))$. 

For all other branches, replace any proof call of the form $(\delta_j,\Psi): S(\delta_j)\Psi$ with $j \neq \gamma$ by $\xi_j \Psi$ to obtain proof $\xi_\gamma^S$. Applications of $\neg_r$ and $c_r$ will then yield a proof of $\mathcal{F}(S_\gamma)\{ m_\gamma \leftarrow y \} \vdash \mathcal{F}(S_\gamma)\{ m_\gamma \leftarrow y +1 \}$.

The desired proof $\xi_\gamma$ is then constructed as follows:

\begin{prooftree}
    \def\fCenter{\mbox{\ $\Rightarrow$\ }}
        \AxiomC{$ \xi_\gamma^B$}
        \RightLabel{$\neg_r^*$}
        \UnaryInfC{$\vdash \mathcal{F}(S_\gamma)\{ m_\gamma \leftarrow 0 \}$}
        \AxiomC{$ \xi_\gamma^S$}
        \RightLabel{$\neg_r^* , c_r^*$}
        \UnaryInfC{$\mathcal{F}(S_\gamma)\{ m_\gamma \leftarrow y \} \vdash \mathcal{F}(S_\gamma)\{ m_\gamma \leftarrow y +1 \}$}
        \RightLabel{$ind$}
        \BinaryInfC{$\vdash \mathcal{F}(S_\gamma)$}
\end{prooftree}
Note that in case of a proof call which includes an instantiation, we use $cut$ instead of $ind$.
Finally, we use cuts to derive $S_\gamma$ from the proof of $\mathcal{F}(S_\gamma)$.

For proof tuples without internal recursion, it suffices to replace any non-self-referencing proof call with its respective \textbf{PRA} derivation; an application of $ind$ is not necessary.

\end{proof}
The introduction of quantifiers does not affect our translation of \textbf{PRA} proofs into proof schemata, as demonstrated in a previous paper. Therefore, we will simply present the proof from \cite{LPAR2024C} here:
\begin{lemma} \label{lemma:pratoschema}
  Let $\pi$ be a \textbf{PRA} proof of $S$. Then there exists a proof schema with end-sequent $S$.

\end{lemma}
\begin{proof}

Let $\pi$ contain $\alpha$ induction inferences

\begin{prooftree}
\def\fCenter{\mbox{\ $\Rightarrow$\ }}

\AxiomC{$ \Gamma_\beta \vdash \Delta_\beta , F_\beta(0) $}
\AxiomC{$\Gamma_\beta, F_\beta(y) \vdash \Delta_\beta, F_\beta(y+1)$}

\RightLabel{$ind$}
\BinaryInfC{$ \Gamma_\beta \vdash \Delta_\beta, F_\beta(n_\beta) $}

\end{prooftree}
where $a \leq \beta \leq \alpha$. W.l.o.g. assume that if $\gamma < \beta$ then the induction inference with conclusion $\Gamma_\beta \vdash \Delta_\beta , F_\beta (n_\beta)$ is above the induction inference with conclusion $\Gamma_\gamma \vdash \Delta_\gamma , F_\gamma (n_\gamma)$.
We define $\vec{n} = \{ n_i \mid i \in \{ 1 ... \alpha \}\} \cup $ $V(\pi)$ as the set of all induction variables, where $n_i$ denotes the induction variable of the i-th induction inference, together with the set of free variables and constants $V$ in $\pi$.
Let $T$ be the transformation taking an \textbf{PRA} proof to a proof schema by replacing the induction inferences with conclusion 
$ \Gamma_\gamma \vdash \Delta_\gamma , F_\gamma (n_\gamma)$ by a proof call 
$ (\delta_\gamma, \{ m \leftarrow n_\gamma\}) : S(\delta_\gamma) \{ m \leftarrow n_\gamma\}$
with $S(\delta_\gamma) = \Gamma_\gamma \vdash \Delta_\gamma , F_\gamma (m)$.

We will inductively construct a proof schema $\mathcal{D} =\{   ( \delta_i, \rho (\delta_i, \vec{n}, 0), \rho(\delta_i, \vec{n}, m +1)) \mid i \in \{ 1 ... \alpha \}          \} $
with end-sequent $S(\delta_{i}) =  \Gamma_{i} \vdash \Delta_{i}, F_{i}(m +1) $ for each tuple and $\delta_i > \delta_{i+1}$. Assume we already constructed proof schema  
$\mathcal{D}_{\beta + 1} = \{   ( \delta_i, \rho (\delta_i, \vec{n}, 0), \rho(\delta_i, \vec{n}, m +1)) \mid i \in \{ (\beta+1) ... \alpha \}          \}$.

Consider the induction inference with conclusion $ \Gamma_\beta \vdash \Delta_\beta, F_\beta(n_\beta)$. Let $\varphi_1$ be the derivation above the left premise and $\varphi_2$ be the derivation above the right premise.
We construct a proof schema $\mathcal{D}_\beta = \{   ( \delta_\beta, \rho (\delta_\beta, \vec{n}, 0), \rho(\delta_\beta, \vec{n}, m +1))\} \cup \mathcal{D}_{\beta+1} $ with $\rho (\delta_\beta, \vec{n}, 0) = T(\varphi_1)$ and 
$\rho(\delta_\beta, \vec{n}, m +1) = $

\begin{prooftree}
\def\fCenter{\mbox{\ $\Rightarrow$\ }}
  \AxiomC{$(\delta_\beta, \emptyset ): S(\delta_\beta)$}
\UnaryInfC{$ \Gamma_\beta \vdash \Delta_\beta, F_\beta(m) $}

  \AxiomC {$T(\varphi_2)$ }
\UnaryInfC{$\Gamma_\beta,  F_\beta(m) \vdash \Delta_\beta, F_\beta(m +1) $}

\RightLabel{$cut, c_l^*, c_r^*$}
\BinaryInfC{$\Gamma_\beta \vdash \Delta_\beta, F_\beta(m +1) $}
\end{prooftree}
Summarising, $( \delta_\beta, \rho (\delta_\beta, \vec{n}, 0), \rho(\delta_\beta, \vec{n}, m +1))$
is a proof schema tuple with end-sequent $\Gamma_\beta \vdash \Delta_\beta, F_\beta(n) $, as desired.

Finally, the part of $\pi$ located beneath the last induction inference is translated into proof schema $\mathcal{D'} = \{(\delta', \rho (\delta', \vec{n}, m), \rho (\delta', \vec{n}, m+1 ) )\} \cup \mathcal{D} $ with $S(\delta') = S$ and $\delta' > \delta_i$ for $i \in \{1 ... \alpha \}$. Let $\varphi$ be the derivation above $S$. As there is no internal recursion in $\delta'$ needed, we only define $\rho (\delta', \vec{n}, m) = $ $\dfrac{T(\varphi)}{S}$.

\end{proof}
As proof schemata are intended to be evaluated for a specific parameter assignment, the above translation does not generally apply to the strong quantification of induction variables. Consider the following \textbf{PRA} derivation, and for simplicity, let $\varphi_1$ and $\varphi_2$ be induction-free:

\begin{prooftree}
          \AxiomC{$(\varphi_1)$}
        \UnaryInfC{$\Gamma \vdash \Delta, F(0)$}

        \AxiomC{$(\varphi_2)$}
        \UnaryInfC{$\Gamma, F(y) \vdash \Delta , F(y+1)$}

        \RightLabel{$ind$}
        \BinaryInfC{$\Gamma \vdash \Delta , F(n)$}
        \RightLabel{$\forall:r$}
        \UnaryInfC{$\Gamma \vdash \Delta , \forall z F(z)$}
\end{prooftree}
A translation, as described in Lemma \ref{lemma:pratoschema}, would initially yield a proof schema $\mathcal{D} = ( \delta, \rho (\delta, \vec{n}, 0), \rho(\delta, \vec{n}, \linebreak m +1))$ with end-sequent $S(\delta) = \Gamma \vdash \Delta, F(m)$ for the induction. In the subsequent step, a strong quantification of parameter $m$ is not feasible, as evaluating $\mathcal{D}$ would violate the eigenvariable condition in $\forall_r$.

However, there are instances where proof schemata can simulate strongly quantified induction variables in \textbf{PRA} proofs. If a quantified induction variable is subsequently cut in a proof, we can omit the application of the quantifier rule and shift the cut upwards in a Gentzen-style manner.
To achieve this, we locate all applications of a weak quantifier rule that lead to the cut formula in the right branch of the proof and instantiate the parameter of the proof schemata representing the induction with the respective terms found in this way. This process is illustrated in the following example.

\begin{example}
In this example we translate a \textbf{PRA} derivation with a strongly quantified induction variable into a proof schema. 

Consider the following \textbf{PRA} derivation $\pi$ with $\varphi_1$, $\varphi_2$, $\psi_1$ and $\psi_2$ induction free for simplicity. Note that the induction variable $n$ of the left induction rule is strongly quantified and subsequently cut. Let $\pi$ be:

\begin{scriptsize}
\begin{prooftree}
    \AxiomC{$(\varphi_1)$}
    \UnaryInfC{$\Gamma \vdash \Delta, F(0)$}

    \AxiomC{$(\varphi_2)$}
    \UnaryInfC{$\Gamma, F(y) \vdash \Delta, F(y+1)$}

    \RightLabel{$ind$}
    \BinaryInfC{$\Gamma \vdash \Delta, F(n)$}
    \RightLabel{$\forall_r$}
    \UnaryInfC{$\Gamma \vdash \Delta, \forall z F(z)$}

%%%%%%

    \AxiomC{$(\psi_1)$}
    \UnaryInfC{$\Sigma, F(t_1) \vdash \Pi, G(0)$}
    \RightLabel{$\forall_l$}
    \UnaryInfC{$\Sigma, \forall z F(z) \vdash \Pi, G(0)$}

    \AxiomC{$(\psi_2)$}
    \UnaryInfC{$\Sigma, F(t_2), G(x) \vdash \Pi, G(x+1)$}
    \RightLabel{$\forall_l$}
    \UnaryInfC{$\Sigma, \forall z F(z), G(x) \vdash \Pi, G(x+1)$}

    \RightLabel{$ind$}
    \BinaryInfC{$\Sigma,\forall z F(z) \vdash \Pi, G(m)$}

    \RightLabel{$cut$}
    \BinaryInfC{$\Gamma, \Sigma \vdash \Delta, \Pi, G(m)$}

\end{prooftree}
\end{scriptsize}
In the right branch of $\pi$, there are two applications of $\forall_l$. Note the respective terms $t_1$ and $t_2$, as we will later use them to instantiate the proof schema representing the induction in the left branch.

We construct a proof schema $\mathcal{D} = \{  (\delta_1, \rho (\delta_i, \vec{n}, 0), \rho(\delta_i, \vec{n}, m_i +1)  ) \mid i \in \{  0,1,2\}  \}$.
For the left induction in $\pi$, we define $\rho(\delta_2, \vec{n}, 0) = \varphi_1$ and $\rho(\delta_2, \vec{n}, m_2 +1)$ as 

\begin{prooftree}
\def\fCenter{\mbox{\ $\Rightarrow$\ }}
    \AxiomC{$(\delta_2, \emptyset ): S(\delta_2)$}
    \UnaryInfC{$ \Gamma \vdash \Delta, F(m_2) $}

    \AxiomC {$(\varphi_2)$ }
    \UnaryInfC{$\Gamma,  F(m_2) \vdash \Delta, F(m_2 +1) $}

    \RightLabel{$cut, c_l, c_r$}
    \BinaryInfC{$\Gamma \vdash \Delta, F(m_2 +1) $}
\end{prooftree}
For the right induction in $\pi$, we omit the applications of $\forall_l$ and use $w_l$ instead. We define $\rho(\delta_1, \vec{n}, 0) $ as 
\begin{prooftree}
    \AxiomC{$(\psi_1)$}
    \UnaryInfC{$\Sigma, F(t_1) \vdash \Pi , G(0)$}
    \RightLabel{$w_l$}
    \UnaryInfC{$\Sigma, F(t_1), F(t_2) \vdash \Pi, G(0)$}
\end{prooftree}
and $\rho(\delta_1, \vec{n}, m_1 +1) $ as
\begin{prooftree}
    \AxiomC{$(\delta_1, \emptyset):S(\delta_1)$}
    \UnaryInfC{$\Sigma, F(t_1), F(t_2) \vdash \Pi, G(m_1)$}

    \AxiomC{$(\psi_2)$}
    \UnaryInfC{$\Sigma, F(t_2), G(m_1) \vdash \Pi, G(m_1 +1)$}
    \RightLabel{$w_l$}
    \UnaryInfC{$\Sigma, F(t_1), F(t_2), G(m_1) \vdash \Pi, G(m_1 +1)$}

    \RightLabel{$cut$, $c_l$, $c_r$}
    \BinaryInfC{$\Sigma, F(t_1), F(t_2) \vdash \Pi, G(m_1 +1)$}
\end{prooftree}
Lastly we define $\delta_0$. As previously mentioned, we instantiate the proof schema $\delta_2$ with the terms $t_1$ and $t_2$. This allows us to perform a $cut$ on an instantiated formula rather than a quantified one, and we can omit the application of $\forall_r$.
Since there is no internal recursion in $\delta_0$, we only define $\rho(\delta_0, \vec{n} ,m_0)$ as
\begin{small}
\begin{prooftree}
    \AxiomC{$(\delta_2, \{ m_2 \leftarrow t_1 \}): S(\delta_2)\{ m_2 \leftarrow t_1 \}$}
    \UnaryInfC{$\Gamma \vdash \Delta, F(t_1)$}

    \AxiomC{$(\delta_2, \{ m_2 \leftarrow t_2 \}): S(\delta_2)\{ m_2 \leftarrow t_2 \}$}
    \UnaryInfC{$\Gamma \vdash \Delta, F(t_2)$}

    \RightLabel{$\land_r$}
    \BinaryInfC{$\Gamma \vdash \Delta, F(t_1) \land F(t_2)$}
    
    %%%%

    \AxiomC{$(\delta_1, \emptyset): S(\delta_1)$}
    \UnaryInfC{$\Sigma, F(t_1), F(t_2) \vdash \Pi, G(m_1)$}
    \RightLabel{$\land :l$}
    \UnaryInfC{$\Sigma, F(t_1) \land F(t_2) \vdash \Pi, G(m_1)$}

    \RightLabel{$cut$}
    \BinaryInfC{$\Gamma, \Sigma \vdash \Delta, \Pi, G(m_1)$}
    
\end{prooftree}
\end{small}
%Note that $\delta_0$ utilizes two instantiations of $\delta_2$, using the terms $t_1$ and $t_2$ found in the right branch of the original $\textbf{PRA}$ proof $\pi$.

\end{example}

\section{Conclusion}

It was shown in \cite{Thesis.Lolic.2020} that when a proof with induction is formulated as proof schema, a recursive structure that represents the proof's Herbrand sequent can be extracted. It however remained an open question to relate proof schemata to systems of arithmetic, i.e. to identify the class of proofs that can be represented as proof schemata and analyzed with the methods in \cite{Thesis.Lolic.2020}. A first investigation in this direction was presented in \cite{LPAR2024C}, where it was shown that proofs in quantifier-free primitive recursive arithmetic (quantifier-free PRA) can be represented as proof schemata, and that quantifier-free proof schemata can be translated back into quantifier-free PRA. In this work we generalize and extend this result by investigating also the cases for quantifiers in proofs. We demonstrate the translation in both directions, with the condition that the inductive parameter is not quantified. 

Together with a completeness result for the proof analysis method in \cite{Thesis.Lolic.2020} (a result not obtained so far), the result in this paper will yield a realization of Herbrand's theorem for an expressive fragment of formal number theory.

%\nocite{*}
\bibliographystyle{eptcs}
\bibliography{refs.bib}
\end{document}